\documentclass{article}

\usepackage{a4wide,graphicx,latexsym}

\newtheorem{theorem}{Theorem}

\newenvironment{proof}{\mbox{\newline\smallskip}\bf{Proof:}\rm{ }}{$\hfill\Box$\newline}

\begin{document}

\title{Google Scholar makes it Hard --\\
the complexity of organizing one's publications}
\author{Hans L. Bodlaender \and Marc van Kreveld}
\date{\normalsize Department of Information and Computing Sciences\\
\normalsize Utrecht University}
\maketitle

\begin{abstract}
With Google Scholar, scientists can maintain their publications on personal
profile pages, while the citations to these works are automatically collected and counted. 
Maintenance of publications is done manually by the researcher herself, and
involves deleting erroneous ones, merging ones that are the same but which were not
recognized as the same, adding forgotten co-authors, and correcting titles of papers
and venues.
The publications are presented on pages with 20 or 100 papers in the web page interface
from 2012--2014.\footnote{Since mid 2014, Google Scholar's profile pages
allow any number of papers on a single page.}
The interface does not allow a scientist to merge two version of a paper 
if they appear on different pages. 
This not only implies that a scientist who wants to merge certain subsets of
publications will sometimes be unable to do so, but also, we show in this note
that the decision problem
to determine if it is possible to merge given subsets of papers is NP-complete.
\end{abstract}

\section{Introduction}

Most researchers in computer science will be familiar with Google Scholar and its abilities
to maintain publications and their citations. Each researcher has his/her own profile
which is shown as a web page with a list of publications.
Google Scholar determines the number of citations to each publication and by default,
lists them in this order on web pages for that researcher. 
Since the collection of the data is automated, it will contain various mistakes, 
many of which are caused by other scientists who fail to give the title or other 
essential information on a paper correctly. 
As a consequence, a single paper may have various versions in the list,
with a slightly different title or publication venue, or with co-authors missing (see Figure~\ref{f:profile}). 
\begin{figure}[htb]
\begin{center}
\includegraphics{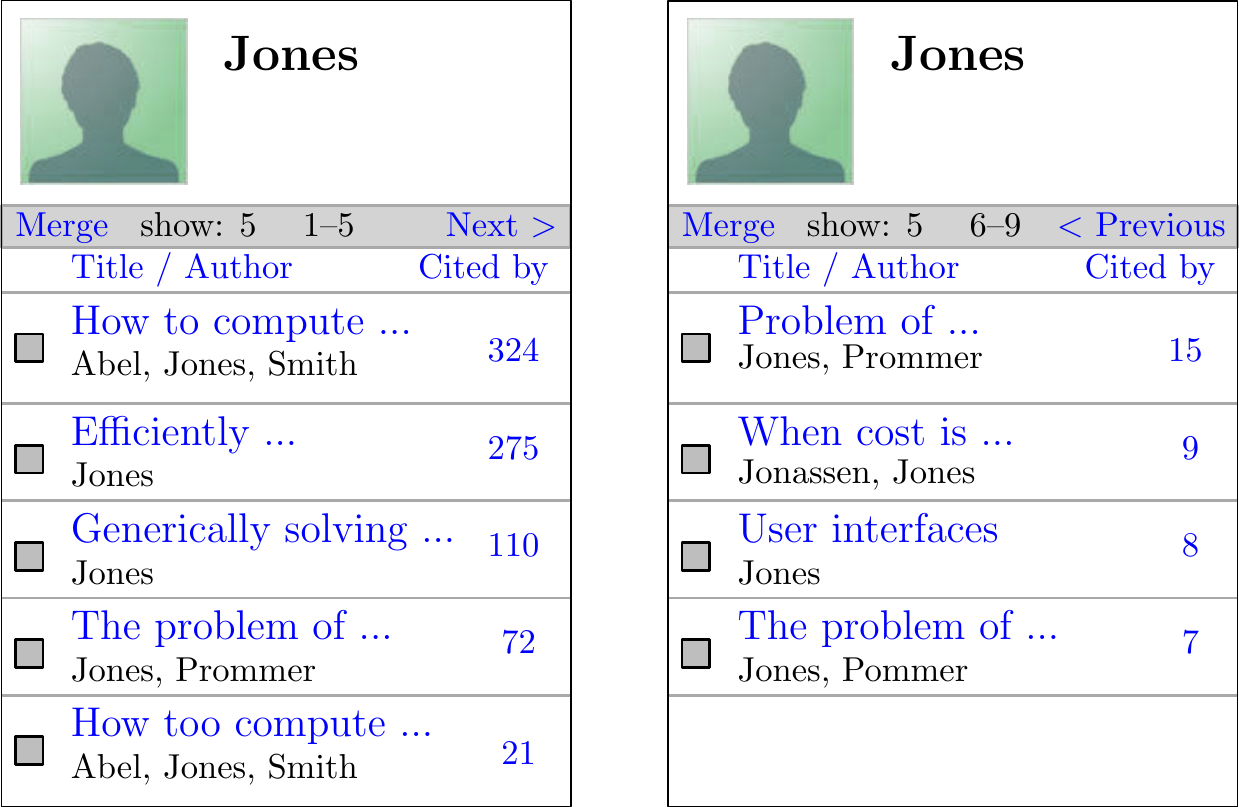}
\end{center}
\caption{Two pages for the author Jones with six different papers occurring as nine versions of papers.
Jones must perform three merges to correct the data, and the order is important.}
\label{f:profile}
\end{figure}
Google Scholar offers researchers the possibility to correct 
these mistakes on their own profile page,
for example by allowing them to merge two paper versions into one. 
This creates one version with the citations of the original versions summed up.
Of course, one could also delete the erroneous
version, but this may cost some citations, which on its turn can influence the 
ever-important H-index and other summary statistics that Google Scholar maintains.

Google Scholar by default shows the publication list on pages with 20 papers. 
It is possible to change this number to 100. Since mid 2014, a change in
the interface makes it possible to get all publications on a single page.
In this note we assume the interface in use from 2012 until mid 2014, when
this was not possible and the maximum was 100 on a single page.
To merge two papers, both should be selected on the web page, after which
the merge action can be executed. However, selection of two papers is possible
only if they appear on the same page, and therefore, merging can be done only if the
two papers appear in the same group of 100 papers. For example, if one
paper is the 103rd by citation count and another paper is the 187th, 
then they can be merged, but if one is the 97th and the other is the 105th, 
then they cannot be merged.

The order in which papers are merged is important. If there are two pairs of papers to be
merged, for example at positions 4 and 12, and at positions 101 and 107, 
then merging the first pair first will move the positions of the latter pair 
to 100 and 106, putting them on different pages. But merging the second pair first
still allows the merging of the first pair. Notice that
the position of a paper can change both forward and backward due to a merge.


Besides the problem that desired merges sometimes cannot be done, the computational
problem of deciding whether a number of merges can be done (and therefore, finding the
correct order) is computationally intractable. This implies that a polynomial-time
algorithm to produce the sequence of merges is unlikely to exist.

\section{A proof of intractability}

To prove intractability, or, NP-completeness of the problem, we will formalize it first.
Let $n$ be the total number of versions of papers initially in a problem instance,
and let $p$ be the page size. Let the paper versions be $v_1,\ldots,v_n$, and assume that
paper version $v_i$ is cited $c(v_i)$ times. 
A problem instance consists of the sequence $c(v_1), \ldots, c(v_n)$, and a partition
of $1,\ldots,n$ into subsets where two or more versions in
the same subset indicates that they are different versions of the same paper, 
and therefore, they are to be merged into one. The Google Scholar Merge Problem is
the problem of deciding whether for every subset, all of its versions can be merged.
When two versions are merged, they appear as one new version and their citations are 
added.
After each merge, the new set of versions appears in sorted order on citation 
count. When citation counts are the same, the papers will appear in some other
well-defined order, but this will be irrelevant for the intractability proof
and we will ignore this issue.

\begin{theorem}
The Google Scholar Merge Problem is NP-complete.
\end{theorem}
\begin{proof}
First, we will verify that the Google Scholar Merge Problem is in NP. 
This is easy: a suggested merge order can easily be checked in quadratic time
or less. 

Second, we use another NP-complete or NP-hard problem and provide a
reduction to our problem, namely {\sc 3-partition}~\cite{Garey1979}.
A {\sc 3-partition} instance
consists of a set $a_1,\ldots,a_{3m}$ of positive integers and an integer $B$,
and asks whether we can partition $a_1,\ldots,a_{3m}$ in $m$ subsets of $3$ integers 
that each sum up to $B$. The integers $a_1,\ldots,a_{3m}$ are all strictly between $B/4$ and $B/2$,
which ensures that any subset that sums to $B$ consists of exactly three integers.

We describe the reduction from {\sc 3-partition} to the Google Scholar Merge Problem. 
First, we double all $a_i$ and $B$ to ensure that they are even.
With slight abuse of notation we continue to use the notation 
$a_1,\ldots,a_{3m}$ and $B$ for these doubled values. 

We set the page size to $3m$.
Let $D$ be some large, even integer; we can use $D=3mB$.

Our instance consists of one paper $P$ with many versions and many papers with one version,
see Figure~\ref{f:initial} for a schematic depiction.

\begin{figure}[htb]
\begin{center}
\includegraphics{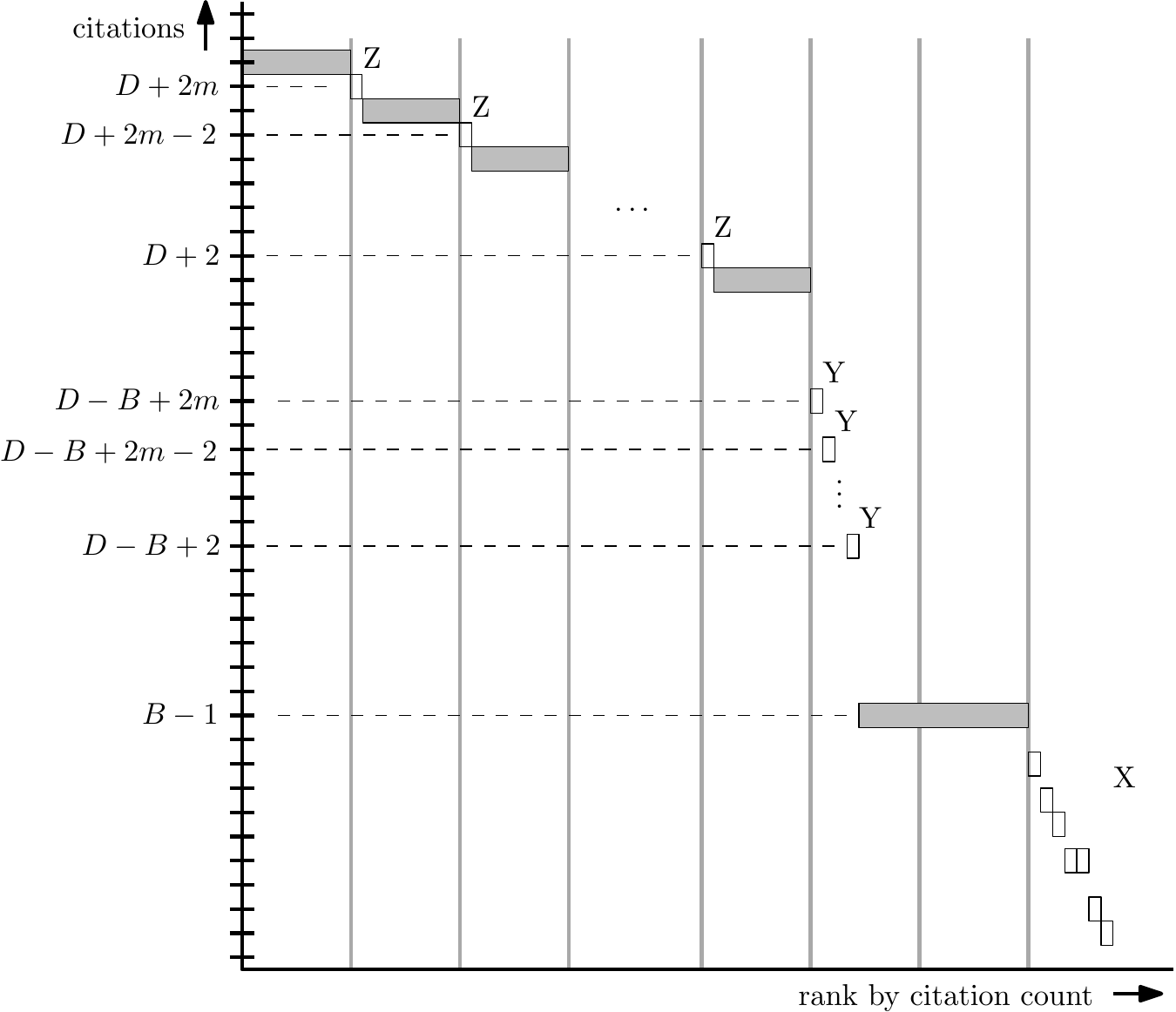}
\end{center}
\caption{Initial situation of the Google Scholar Merge problem arising 
from a {\sc 3-partition} instance. Vertical grey lines are page boundaries.}
\label{f:initial}
\end{figure}

First, we take one paper $P$ with $5m$ versions that are cited:
\begin{itemize}
\item $a_1$, \ldots, $a_{3m}$ times (type X),
\item $D - B + 2i$ times (type Y), for each $i$, $1\leq i \leq m$, 
\item $D +2i$ times (type Z), for each $i$, $1 \leq i \leq m$.
\end{itemize}
All versions of $P$ have an even number of citations, so any merge of these will
also have an even number of citations.

Second, we take many papers with only one version as follows:
\begin{itemize}
\item for each $i$, $1\leq i \leq m$, we take $3m-1$ papers with one version,
each with $D +2i -1$ citations,
\item we take $3m$ papers with $D+2m+1$ citations
\item we take $5m$ papers with $B-1$ citations
\end{itemize}

Call the latter the {\em single papers}; they all have an odd number of citations
and will not be merged.
A group of at least $3m-1$ single papers with the same 
number of citations is called a {\em block}; there are $m+2$ blocks in the instance.
Since a block has at least as many papers as the page size minus one, 
paper versions on different
sides of a block in citation count can never be on the same page. 
Notice that the type Z versions of paper $P$ are separated by the blocks.
The only way to ``rescue''
a type Z version and get it out of the adjacent blocks is to create a version
of paper $P$ that has exactly the same number of citations. The construction
makes sure that this can only be accomplished using exactly one type Y version and
three type X versions of paper P. 

We claim that this Google Scholar Merge Problem instance has a solution if and only 
if the {\sc 3-partition} instance has a solution.

Suppose we have a partition of $a_1,\ldots,a_{3m}$ into $m$ triples that 
each sum up to $B$.
Let $T_i$ be the subset of integers in the $i$th triple. 
For $i$ from $1$ to $m$ do the following. 
Merge the three type X versions with $a_j$ citations
for the $a_j \in T_i$ with each other. Note that all intermediate merged versions 
have less than $B-1$ citations together, and thus will appear on the last page.  
When we merged the triple we have a version with $B$ citations, which will
appear on the same page as all type Y versions.
Next, we merge any merged triple with a type Y version with $D-B+2i$
citations, giving it $D+2i$ citations. Then we merge it with the type Z
version that has $D+2i$ citations. Any merge of a type Z, a type Y, and three
type X versions will appear on the first page because it will have more citations
than the block with most citations.
We finish by merging all versions of paper $P$ on the first page.
Hence, we can merge all versions of paper P if the {\sc 3-partition} 
instance has a solution.
\medskip

Conversely, suppose that we can merge all versions of paper P in the instance.
This implies that all type X versions can be merged with other type X versions
into versions with at least $B$ citations, otherwise they would stay after
the block with $B-1$ citations. We can make $\leq m$ of these merged versions,
because the type X versions have $mB$ citations in total. These merged versions
consist of at least three type X versions, because any two type X
versions merged have at most $B-2$ citations 
(since type X versions are cited $a_i$ times with $B/4<a_i<B/2$, for all $1\leq i \leq 3m$).

Furthermore, the assumption that we can merge all versions of paper P implies
that at some point in the merge sequence, we must have had merged versions 
with $D+2$, $D+4$, \ldots, $D+2m$ citations, otherwise the type Z 
versions cannot get to a common page. 

Two type Y versions merged will have at least $2D-2B+6>D+2m$ citations,
so we can use at most one type Y paper if we want to get a version with
as many citations as some type Z version.
To get the type X versions merged with any of the type Z versions, they must
at some point be in a merged version with at least $D+2$ citations, but all
of the type X versions together have only $mB<D+2$ citations. So a type Y
version is always needed in such a merged version, and we conclude that
exactly one type Y version is needed in any merged version to make a merged
version with the same number of citations as any of the type Z versions. 

This implies that we need
at least $m$ merged versions of type X versions, otherwise we do
not have enough versions to merge with the type Y versions to reach
the same citation counts as the type Z versions, so we must be able to make
exactly $m$ merged versions of type X papers. To realize this we cannot 
have a merged version with four type X versions,
because this merged version would have $>B$ citations, and then by
the pigeonhole principle at least one of the $m$ merged versions has only two
versions and hence $<B$ citations. 
It follows that the $m$ merged versions of type X papers 
are triples and they sum up to exactly $B$, showing that the {\sc 3-partition} 
instance has a solution.
This completes the proof of NP-completeness.
\end{proof}

It would have been possible to construct a similar but slightly easier proof 
based on a reduction from the NP-complete problem {\sc Partition}. 
The reason we chose {\sc 3-partition} is that {\sc 3-partition} is NP-complete 
\emph{in the strong sense}~\cite{Garey1979}, which implies that even if the total number
of citations is bounded by a polynomial in $n$ (the integers of the input), then
{\sc 3-partition} is still NP-complete. A proof based on {\sc Partition} would
require exponentially many citations to the papers of a researcher.

\section{Conclusions}

It is interesting to see an example of a user interface where the user has to
solve an NP-complete problem due to the choice of interface. By comparison,
operations research companies often provide a user interface where a user makes 
a schedule for a scheduling problem that is usually NP-hard, but here the user 
interface assists in the scheduling, and no user interface would make the problem easier.
Other examples are certain puzzle games where large-instance extensions of the 
puzzle type are NP-complete or even PSPACE-hard~\cite{Hearn2009}. 
Here the idea is clearly that the task is supposed to be hard.
In the Google Scholar Merge Problem, the NP-completeness is accidental, and
arises from the way the user interface is designed. 
Google Scholar makes it hard for researchers to maintain their Google Scholar
profile pages. Since the interface change in the middle of 2014, the problem
is no longer intractable because there is now a button ``show more'' that allows
a user to get more papers on the same page, eventually allowing all papers to be on
a single page.


It is open to determine whether the Google Scholar Merge Problem is still NP-complete if:
\begin{itemize}
\item
the page size is bounded by a constant,
\item
the number of versions of any paper is bounded by a constant, or
\item
the page start can be placed at any rank number, not just at multiples of the page size plus one.
\end{itemize}

\bibliographystyle{plain}
\bibliography{citeproblem}

\end{document}